\documentclass[conference]{IEEEtran}
\IEEEoverridecommandlockouts
\usepackage{cite}
\usepackage{amsmath,amssymb,amsfonts}
\usepackage{amsthm}
\usepackage[ruled]{algorithm2e}
\usepackage{algorithmic}
\usepackage{graphicx}
\usepackage{textcomp}
\usepackage{bbm}
\usepackage{xcolor}
\def\BibTeX{{\rm B\kern-.05em{\sc i\kern-.025em b}\kern-.08em
    T\kern-.1667em\lower.7ex\hbox{E}\kern-.125emX}}

\newcommand{\R}{{\mathbb{R}}}
\newcommand{\bbP}{{\mathbb{P}}}
\newcommand{\bbI}{{\mathbbm{1}}}

\newtheorem{theorem}{Theorem}

\newtheorem*{remark}{Remark}

\begin{document}


\title{Optimal Transport for Super Resolution Applied to Astronomy Imaging}

\author{\IEEEauthorblockN{Michael Rawson}
\IEEEauthorblockA{\textit{Department of Mathematics} \\
\textit{University of Maryland at College Park}\\
Maryland, USA \\
rawson@umd.edu}
\and
\IEEEauthorblockN{Jakob Hultgren}
\IEEEauthorblockA{\textit{Department of Mathematics} \\
\textit{University of Maryland at College Park}\\
Maryland, USA \\
hultgren@umd.edu}}

\maketitle

\begin{abstract}
Super resolution is an essential tool in optics, especially on interstellar scales, due to physical laws restricting possible imaging resolution. We propose using optimal transport and entropy for super resolution applications. We prove that the reconstruction is accurate when sparsity is known and noise or distortion is small enough. We prove that the optimizer is stable and robust to noise and perturbations. We compare this method to a state of the art convolutional neural network and get similar results for much less computational cost and greater methodological flexibility.  
\end{abstract}

\begin{IEEEkeywords}
optimal transport, Wasserstein distance, super resolution, compressed sensing, sparse imaging, sparse regularization, sparsity, maximum entropy, convolutional neural network
\end{IEEEkeywords}




\section{Introduction and Background}

\subsection{Super Resolution}

Super resolution seeks to improve image resolution without further data collection. This is useful when important features or pixels are missing. Improving the measurement device, such as a camera or telescope, will improve the image resolution but only up to limits governed by physical laws, for example the diffraction limit. Super resolution can increase image resolution beyond this point given constraints that give a well-posed inverse problem. The most common constraint is that the true image is either sparse or smooth in some basis. For example, Gaussian noise is a common input that blurs important features. Removing additive Gaussian noise can be done, imperfectly, by solving an inverse problem that constrains total variation and hence enforces smoothness \cite{ng}. 

A more general solution is to minimize with respect to a regularizing term that maximizes sparsity. Compressed sensing methods often minimize an objective function 
involving the $L^1$ norm of the solution \cite{candes}. Minimizing $L_0$ maximizes sparsity but $L^1$ is usually used instead for its convex properties. Another regularizer that maximizes sparsity is entropy \cite{gull, cornwell1985}. Neural networks or deep learning has more recently been used for inverse problems, especially on images \cite{ongie}.

\subsection{Optimal Transport}
Optimal transport has a long history in pure and applied math, and has been used recently in the field of imaging \cite{cuturi}. One starts with a source distribution, $\mu\in\bbP(X)$, target distribution, $\nu\in\bbP(Y)$, and a cost $C:X \times Y \rightarrow \R$ on spaces $X$ and $Y$. In the discrete setting relevant to scientific computing (where $X$ and $Y$ are finite sets), distributions are represented as finite dimensional vectors.
Given two positive vectors of $L^1$-norm equal 1, $\mu\in\R^n_{>0}, \nu \in \R^m_{>0}$, together with a cost $C$ represented as an $n\times m$ matrix $(C_{ij})$, the optimal transport plan between $\mu$ and $\nu$ is 
\begin{equation} \label{eq:OTPlan} 
    \arg\min_{P\in \Pi(\mu,\nu)} \sum_{i=1}^n \sum_{j=1}^m C_{ij} P_{ij} 
\end{equation}
where $\Pi(\mu,\nu)$ is the set of transport plans from $\mu$ to $\nu$:
$$ \Pi(\mu,\nu) = \{P\in\R^{n\times m}: \sum_{i=1}^n P_{ij} = \nu_j\ \forall j, \sum_{j=1}^m P_{ij} = \mu_i\ \forall i\}. $$
When $m=n$ and $C$ defines a metric on $\{1,\ldots,n\}$ (i.e. $d(i,j) := C_{ij}$ is symmetric, non-degenerate, and satisfies the triangle inequality) then the minimum value
\begin{equation} \label{eq:Wdist} 
    d_W(\mu,\nu) := \min_{P\in \Pi(\mu,\nu)} \sum_{i=1}^n \sum_{j=1}^m C_{ij} P_{ij} 
\end{equation}
defines a metric on $\bbP(\{1,\ldots,n\})$ often referred to as the 1st Wasserstein metric or just the Wasserstein metric \cite{villani}.

\subsection{Entropy}\label{sec:Entropy}
For a probability mass function $p:\mathcal J\rightarrow \R$ (non-negative and sums to 1), we will use $H(p)$ to denote its entropy,
\begin{equation} 
    H(p) = -\sum_{\iota \in \mathcal J} p(\iota) \ln(p(\iota)) \label{eq:Entropy}. 
\end{equation}
As we represent probability densities with vectors and matrices, we will consider the indices to be in the domain $\mathcal J$.
In \eqref{eq:Entropy}, we use the convention that $0\cdot\ln(0) = 0$. With this convention, $H$ defines a continuous non-negative function on the probability simplex, differentiable in the interior of the probability simplex and with a derivative unbounded at the boundary. The key idea is that sparse arrays have low entropy.

\subsection{The Sinkhorn algorithm}\label{sec:Sinkhorn}

The distance $d_W$ in Equation \eqref{eq:Wdist} 
can be computed exactly using methods from linear programming. However, for large $n$ the quickly growing computation times excludes this from many applications \cite{cuturi}. In applications involving large data sets, $d_W$ is often replaced by its \emph{entropic regularization}, which is more feasible from a computational perspective \cite{peyre}. Given a small constant $\epsilon>0$, the $\epsilon$-regularized distance between $\mu,\nu\in \bbP(\{1,\ldots,n\})$ is 
\begin{equation} \label{eq:Wepsilon}
    d_W^\epsilon(\mu,\nu) := \sum_{i=1}^n \sum_{j=1}^m C_{ij}P^*_{ij}
\end{equation}
where 
\begin{equation}
    P^* = \arg\min_{P\in \Pi(\mu,\nu)} \sum_{i=1}^n \sum_{j=1}^m C_{ij}P_{ij} - \epsilon H(P).
    \label{eq:WepsilonObj}
\end{equation}
The regularized objective function in \eqref{eq:WepsilonObj} is strictly convex and proper, hence always admits a unique minimizer. While the minimizer in the true optimal transport problem is usually very sparse, the entropy term in \eqref{eq:WepsilonObj} pushes the minimizer away from the boundary of the unit simplex, producing a less sparse minimizer. As $\epsilon\rightarrow 0$, this minimizer converges to a minimizer of \eqref{eq:Wdist} (the minimizer with highest entropy if there are more than one), and $d_W^\epsilon$ converges to $d_W$ \cite{peyre}. 

A simple application of Lagrange multipliers show that the minimizer of \eqref{eq:WepsilonObj} is the unique element in $\Pi(\mu,\nu)$ on the form
\begin{equation} P_{ij} = \sum_{i=1}^n \sum_{j=1}^m f_i e^{-C_{ij}/\epsilon} g_j \label{eq:ArgmindWepsilon} \end{equation}
for some unknown positive multipliers $f=(f_1,\ldots,f_n),\ g=(g_1,\ldots g_n)$ \cite{peyre}. Determining $f$ and $g$ from $\mu, \nu$ and the matrix $(e^{-C_{ij}/\epsilon})$ is known as the matrix scaling problem, and a standard algorithm to find approximate solutions is the iterative proportional fitting procedure, also known as the Sinkhorn Algorithm \cite{peyre}. The matrix \eqref{eq:ArgmindWepsilon} lies in $\Pi(\mu,\nu)$ if $\sum_i P_{ij} = \mu_i$ for all $i$ and $\sum_j P_{ij} = \nu_j$ for all $j$. The Sinkhorn algorithm proceeds iteratively, alternating between updating $f$ so that the first of these conditions is satisfied and updating  $g$ so that the second of these conditions is satisfied (see Algorithm~\ref{algo:Sinkhorn}).

\subsection{Wasserstein Distance Gradient}
The multipliers $f\in\R^n$ and $g\in\R^m$ in the previous section can be thought of as dual variables for the optimization problem \eqref{eq:WepsilonObj}. More precisely, 
$$ F = \epsilon \ln(f), \; G = \epsilon \ln(g) $$ 
are the Lagrange multipliers for \eqref{eq:Wepsilon} \cite{cuturi}. When considering the regularized distance between $\mu$ and $\nu$ as a function of $\mu$ (keeping $\nu$ fixed) this can, at least for positive $\mu$ and $\nu$, be exploited to approximate its gradient (see for example \cite{frogner, lellman}). The Sinkhorn Algorithm, used to approximate $d_\mu^\epsilon(\mu,\nu)$ and its gradient with respect to $\mu$ is summarized in Algorithm~\ref{algo:Sinkhorn}. Note that the output $F$ and $G$ of Algorithm~\ref{algo:Sinkhorn} needs to be projected onto the tangent space of the probability simplex to yield an approximation of the true gradients. 

\begin{algorithm}
    \textbf{Input:} 
    
    \quad $\mu, \nu \in \R^n$ : positive probability vectors 

    \quad $C\in\R^{n \times n}$ : cost matrix 
    
    \quad $\epsilon$ : positive regularization parameter
    
    \textbf{Output:}
    
    \quad $d_W^\epsilon \in \R$ : regularized distance between $\mu$ and $\nu$
    
    \quad $F \in \R^n$ : gradient of $d_W^\epsilon(\mu,\nu)$ with resp. to $\mu$ at $\mu$
    
    \quad $G \in \R^n$ : gradient of $d_W^\epsilon(\mu,\nu)$ with resp. to $\nu$ at $\nu$ 

    \textbf{Begin:} 

    $f = (1,\ldots,1)\in \R^n$
    
    $g = (1,\ldots,1)\in \R^n$
    
	\While{$f$ and $g$ have not converged}
 	{
 	    \For{$1 \le i \le n$}{
     	$f_i = \mu_i / \left(\sum_j \exp(-C_{ij}/\epsilon) g_j\right)$ 
     	}
     	\For{$1 \le j \le n$}{
     	$g_j = \nu_j / \left(\sum_i \exp(-C_{ij}/\epsilon) f_i\right)$ 
     	}
  	}
  	
    $d_W^\epsilon = \sum_{i=1}^n \sum_{j=1}^m f_i g_j \exp(-C_{ij} /\epsilon) C_{ij}$ 
  	
  	$F = -\epsilon \ln(f)$
  	
  	$G = -\epsilon \ln(g)$\\
  	
	\caption{The Sinkhorn Algorithm for Regularized Optimal Transport Distances}
	\label{algo:Sinkhorn}
\end{algorithm}

\begin{remark}
In Algorithm~\ref{algo:Sinkhorn}, if one is only interested the regularized distance $d_W^\epsilon$, the assumption of positivity of $\mu$ and $\nu$ can be relaxed to non-negativity. 
However, reflecting the fact that the entropy term in \eqref{eq:Wepsilon} pushes the minimizer away from the boundary, any entry $F_i$ in $F$ will return as $+\infty$ if $\mu_i=0$. 
\end{remark}

\section{Wasserstein Inverse Problem for Super Resolution}

We propose two super resolution inverse problems that produce sparse solutions which are near to the measurement in Wasserstein distance. For a measurement $\nu$ and positive regularization parameters $\lambda$ and $\lambda'$, we define the \emph{sparse approximation} of $\nu$ as a minimizer 
\begin{equation}
    \mu_* = \arg\min_{\mu \in \bbP(X)} d_W^\epsilon(\mu,\nu) + \lambda H(\mu).
    \label{eq:sparseNeigh}
\end{equation}
and the \emph{sparse retrieval} of $\nu$ as a minimizer 
\begin{equation}
    \mu_* = \arg\min_{\mu: d_W(\mu,\nu)<\lambda'} H(\mu).
    \label{eq:sparseRetriev}
\end{equation}
Problem \eqref{eq:sparseNeigh} and \eqref{eq:sparseRetriev} are essentially dual, and for generic data $\nu$ there is a mapping $\lambda \mapsto \lambda'(\lambda)$ such that $\mu$ is a solution to \eqref{eq:sparseNeigh} if and only if $\mu$ is a solution to \eqref{eq:sparseRetriev}. We will approach \eqref{eq:sparseRetriev} from a theoretical perspective in Section~\ref{sec:MainResults} but use \eqref{eq:sparseNeigh} in our application since it fits well into a gradient descent method. 

At least one minimizer exists by compactness of the finite dimensional probability simplex. The entropy term in \eqref{eq:sparseNeigh} favors sparse solutions. Naturally, there is a trade-off between sparsity of the solution and proximity to the measurement. How these two objectives are prioritized is governed by $\lambda$. For $\lambda=0$, no priority is given to the goal of sparsity and $\mu_* = \nu$. As $\lambda$ increases, $\mu_*$ turns into an increasingly sparse approximation of $\nu$ and when $\lambda \rightarrow \infty$, $\|\mu_*\|_0 \rightarrow 1$.  

This inverse problem is useful whenever there is a natural distance, or cost function, on the index set of $\nu$. If, for example, $\nu$ is given in Fourier space and each entry $\mu_i$ corresponds to a frequency $\sigma_i$, then two natural choices for the cost $C_{ij}$ are $C_{ij} = |\sigma_i-\sigma_j|$ and $C_{ij} = |\ln(\sigma_i/\sigma_j)|$. In the application we describe below, each entry in $\nu$ describes the intensity of a pixel in a $32\times 32$ image and $C_{ij}$ is chosen as the $L^2$-distance between the $i^{th}$ pixel and $j^{th}$ pixel.

\begin{remark}
This method can be contrasted to maximum entropy methods in statistical physics, where the probability distribution with highest entropy (under constraints dictated by observations) is chosen as the best representative of the current state of knowledge about a system. In our context, we work with the crucial assumption of sparsity, which motivates minimizing the entropy instead of maximizing it. 
\end{remark}

\section{Main Results}\label{sec:MainResults}
We will let $\nu\in P(\{1,\ldots,n\})$ be a sparse signal (i.e. $\|\nu\|_0<n$ is small), and use $\bar\nu$ to denote this signal with noise and distortion. In our application, we are interested in determining the support of $\nu$ (i.e. the indices of all non-zero entries in $\nu$) from $\bar\nu$. For two probability vectors $\mu$ and $\nu$ we will say that $\mu$ identifies the structure of $\nu$ if they have the same support, i.e. if $\mu_i>0$ if and only if $\nu_i>0$ for all $i$. Our main theorem (Theorem~\ref{thm1} below) shows that the minimizer of \eqref{eq:sparseNeigh} identifies the structure of $\nu$ under the assumptions that $\nu$ is sparse and the noisy signal $\bar\nu$ is close to $\nu$ in optimal transport distance. As is indicated by Theorem~\ref{thm:NoiseBound} below, the latter assumption is natural when dealing with Gaussian noise since the optimal transport distance, unlike total variation and $L^p$ distances, take the geometry of the space into account. 

\subsection{Optimal Transport Bound on Gaussian Noise}
Let the probability distribution $\nu:=\frac{1}{k}\sum_{i=1}^k \delta_{p_i}$ be a sparse signal in $\bbP(\R^d)$ where $\delta$ is the Dirac delta. Assume the noisy signal $\tilde\nu$ is produced in the following way: For each $p_i$ in the sum above, we sample $n$ points $x_i^1,\ldots,x_i^n$ in $\R^d$ according to a normal distribution centered at $p_i$ with independent components of variance $\sigma^2$. Let $N=kn$ be the number of points sampled and $\tilde\nu=\frac{1}{N}\sum_{i=1}^n \sum_{j=1}^m\delta_{x_i^j}$ be the noisy signal. 
\begin{theorem}
\label{thm:NoiseBound}
    Given a sparse signal $\nu\in\bbP(\R^d)$ giving rise to a noisy signal $\tilde\nu$ as described above, the optimal transport distance between $\nu$ and $\tilde\nu$ is bounded by $\frac{\sigma^2}{N}X_N$ where $X_N$ is a random variable with distribution $\chi^2_{dN}$. In particular, the expected value and variance of $\frac{\sigma^2}{N}X_N$ are $d\sigma^2$ and $2d\sigma^4/N$, respectively. 
\end{theorem}
\begin{proof}
    The optimal transport cost is bounded from above by the cost of the transport plan sending each $x_i^j$ to $p_i$. The cost of this plan is 
    $\frac{1}{N}\sum |x_i^j-p_i|^2 $. By assumption, each term in this sum is the squared sum of $d$ normal distributed random variables with mean 0 and variance $\sigma^2$. 
\end{proof}

\subsection{Reconstructing the Support of a Sparse Signal}
\begin{theorem}
    \label{thm1}
    Assume $\nu$ is a sparse signal and $\bar\nu$ is a noisy signal such that 
    $d_w(\nu,\bar\nu) < \delta $. Then the solution of 
    \begin{equation} \label{eq:dualConstraint} \mu=\arg\min_{\mu : d_W(\bar\nu,\mu) \le \delta} H(\mu) \end{equation} 
    will identify the structure of $\mu$, i.e. have the same support as $\mu$, if $||\nu||_0 \leq ||\mu||_0$ for all $\mu$ such that $d_W(\mu,\bar\nu)<2\delta$, with equality only if $\mu$ and $\nu$ has the same support.  
\end{theorem}

\begin{remark}
    The conditions in Theorem~\ref{thm1} can be summarized as a low enough noise level $\delta$ and enough sparsity of the true signal $\nu$ (making it a local minimizer of the $L^0$-norm). It is interesting to note that these conditions are essentially necessary: if the inequality in Theorem~\ref{thm1} is violated by some $\mu$ closer than $\delta$ to $\bar\nu$, then the solution of \eqref{eq:dualConstraint} does not identify the structure of $\nu$. 
\end{remark}

\begin{remark}
    Noise is high entropy, hence it is expected that the noise can be removed by minimizing the entropy. However, if the signal-to-noise ratio is too low, this reconstruction is underdetermined.
\end{remark}

\begin{proof}[Proof of Theorem \ref{thm1}]
    By the triangle inequality, the feasible set in \eqref{eq:dualConstraint} is contained in the ball centered at $\bar\nu$ of radius $2\delta$. As the feasible set in \eqref{eq:dualConstraint} contains $\nu$, this means any solution of \eqref{eq:dualConstraint} has to be $\nu$ or have the same support as $\nu$.
\end{proof}

\begin{theorem}
    Fix a positive probability vector $\nu\in\R^d_{>0}$ such that all elements of $\nu$ are distinct. 
    Then the sparse recovery is continuous to perturbations around $\nu$ for small $\lambda$, i.e. for every $\epsilon'>0$ there exists $\delta>0$, such that 
    if $d_W(\nu,\nu')<\delta$,
    
    $\mu_* = \arg\min_{\mu \in \bbP(X) : d_W(\mu, \nu)<\lambda} H(\mu)$,
    and 
    
    $\mu_*' = \arg\min_{\mu \in \bbP(X) : d_W(\mu, \nu')<\lambda} H(\mu)$
    then $\|\mu_*-\mu_*'\| < \epsilon'.$
    \label{thm:robust}
\end{theorem}

\begin{proof}
    The assumption on $\nu$ guarantees that minimizers are unique for small $\lambda$. Continuity of the minimizer then follows from smoothness of $H$. 
\end{proof}

\begin{algorithm}[h]
    \textbf{Input:} 
    
    \quad $X \in \R^{N \times m \times m}$ : $N$ images size $m \times m$

    \quad $\lambda \in \R$ : positive noise level 

    \quad $0 < \epsilon < 1$ : optimal transportation regularization 

    \quad $C\in\R^{m^2 \times m^2}$ : cost matrix 

    \quad $ J_{\lambda,\epsilon}(x,v) := d_W^\epsilon(x,v) + \lambda H(v) $

    \textbf{Output:}
    
    \quad $K \in \R^N$ : star cluster classification
    
    \textbf{Begin:}
    
    $K = 0$
    
	\For{$i=1,2,...,N$}
 	{
        $ v = X_n;\ w = 1 $ \\
    	\While{ $v$ has not converged }
 	    {
$w = \nabla d_W^\epsilon(X_n,\cdot)|_v + \lambda \nabla H|_v $ \\
$ w = w - \langle w, \frac{1}{m} \bbI \rangle \cdot \frac{1}{m} \bbI $ \\
$ \alpha = \sup\{\alpha \in \R : J_{\lambda,\epsilon}(v) > J_{\lambda,\epsilon}(v - \alpha w) \}$ \\
            $ \alpha = \min\{0.01, \alpha\}$\\
 	        $ v = v - \alpha w $;
 	        $ v = diag(\bbI_{v > 0}) \ v $;
 	        $ v = v/\|v\|_1 $
 	    }
 	    $V_i = v$;
     	$\delta = \max V_i$
     	
        \If{ $rank(H_0(V_i^{-1}([0.75 \delta,\ \delta]))) == 1$ }{ 
            $ K_i = 1 $ 
        }
  	}
	\caption{Optimal Transport Star Cluster Prediction}
	\label{algo:OT}
\end{algorithm}

\section{Minimizing the objective}

We solve \eqref{eq:sparseNeigh} using a gradient descent method with variable step size. More precisely, letting $J(\mu) := d_W^\epsilon(\mu,\nu) + \lambda H(\mu)$ be the objective we set the step size to $\alpha_* := \sup\{\alpha>0 : J(\mu-\alpha \nabla J|_\mu) < J(\mu)\}$. 
As mentioned in Section~\ref{sec:Entropy}, the entropy is not differentiable on the boundary of the probability simplex. 
An effect of this is that the output $F$ in Algorithm~\ref{algo:Sinkhorn} is infinity in all indices where $\mu$ is zero. 
We circumvent this problem by defining the $i$'th entry in the gradient of $J$ to be $0$ whenever $\mu_i=0$. Geometrically, this means that whenever the algorithm reaches a sub-simplex of the probability simplex, it ignores the component of the gradient orthogonal to this sub-simplex, thus remaining in this sub-simplex for the rest of the algorithm. 
Gradient descent will converge to a local minimum on the compact probability simplex since the objective is smooth when restricted to the local simplex face. Algorithm \ref{algo:OT} contains the pseudocode and also makes the star cluster classification described in Section \ref{sec:Star}.

\section{Simulation}

\begin{figure}[h]
\centerline{\includegraphics[width=.5\textwidth]{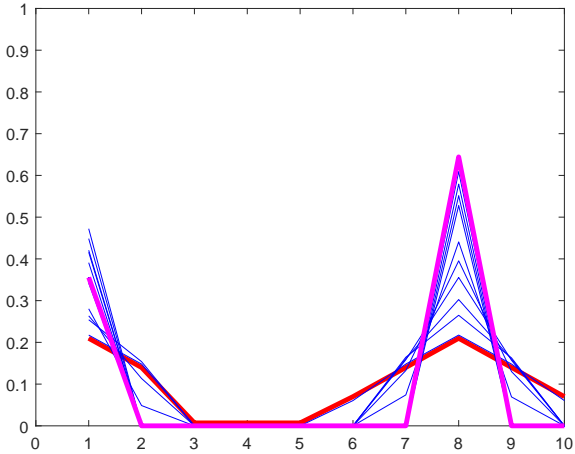}}
\caption{Plot of super resolution O.T. method Algorithm \ref{algo:OT}. Red line is initial distribution. Blue lines are steps along gradient of Equation \eqref{eq:sparseNeigh}. Pink line is final, converged distribution. $\lambda=10$.
epsilon = 0.1. Max Sinkhorn iterations = 5000. Gradient step size = 0.01. Gradient steps=50.}
\label{fig:low_sparse}
\end{figure}

We first show this method's results on a low dimensional example. For example, let the measurement be \\
$ \nu = (0.2, 0.15, 0, 0, 0, 0.1, 0.15, 0.2, 0.15, 0.1)^T. $ 
With sparsity parameter $\lambda = 10$  the method produces the sparse approximation $(0.35,0,0,0,0,0,0,0.65,0,0)$.
This reflects the fact that $\nu$ has two peaks, one peak centered at position 1 and one peak centered at position 8, and that 35\% of the mass of $\nu$ is situated close to position 1 and 65\% of the mass of $\nu$ is situated close to position 8. Figure \ref{fig:low_sparse} plots $\nu$, the gradient descent steps, and the final result.
With sparsity parameter $\lambda = 100$, the method produce the sparse approximation $(0,0,0,0,0,0,0,1,0,0)$, reflecting the fact that most of the mass of $\nu$ is part of a peak centered at position 8. Figure \ref{fig:high_sparse} plots $\nu$, the gradient descent steps, and the final result.

\begin{figure}[t]
\centerline{\includegraphics[width=.5\textwidth]{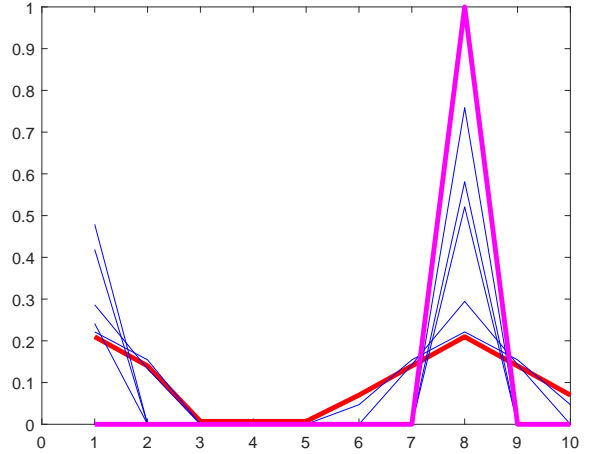}}
\caption{
Plot of super resolution O.T. method Algorithm \ref{algo:OT}. Red line is initial distribution. Blue lines are steps along gradient of Equation \eqref{eq:sparseNeigh}. Pink line is final, converged distribution. $\lambda=100$.
epsilon = 0.1. Max Sinkhorn iterations = 5000. Gradient step size = 0.01. Gradient steps=50.}
\label{fig:high_sparse}
\end{figure}

\section{Star Clustering Application}\label{sec:Star}

\begin{table}[h]
\caption{Confusion matrix of O.T. metod Algorithm \ref{algo:OT} on LEGUS data compared to StarcNet \cite{perez}. Column gives StarcNet classification and row gives Algorithm \ref{algo:OT} classification.}
\begin{center}
\begin{tabular}{c|c|c|}
 & StarcNet Cluster & StarcNet Not Cluster \\
\hline
O.T. Cluster     & 25\% (32) & 13.3\% (17)  \\
\hline
O.T. Not Cluster & 12.5\% (16) & 49.2\% (63)  \\
\hline
\end{tabular}
\label{table:2}
\end{center}
\end{table}

The formation and evolution of star clusters provide insight into the processes governing the birth of stars as well as the dynamical evolution of galaxies \cite{perez}. 
In order to save human hours and get reproducible results, it is of interest to algorithmically detecting star clusters in images of sky patches. Many methods have been proposed to algortihmically detect star clusters, including CLEAN \cite{hogbom}, Multiscale CLEAN \cite{cornwell2008}, IUWT-based CS \cite{li}, decision trees \cite{grasha} and optimal sheaves \cite{robinson}. The state of the art method trains a convolutional neural network (CNN) to classify each sky patch or region in an image as containing a star cluster or not \cite{perez}. These neural networks are notoriously computationally expensive, sensitive to noise, and inflexible to appending or removing data variables. 

We propose using the Wasserstein inverse problem \eqref{eq:sparseNeigh} to detect star cluster locations. 
Our dataset consists of measurements from the Hubble Space Telescope in the survey Treasury Project LEGUS (Legacy ExtraGalactic Ultraviolet Survey) \cite{calzetti}. This data set consists of $32\times 32$ pixel images of star patches. Each image comes in 5 frequency bands (NUV, U, B, V, and I) \cite{calzetti}. We encode each of these in a probability vector $\nu$ where each entry correspond to the intensity of a pixel, normalized to sum to 1. Algorithm~\ref{algo:OT} produces a sparse approximation each image which is classified as a star cluster if it contains just one `peak'. Specifically, the sparse image is made binary (1 and 0) by thresholding at 75\% of the max of the image. Then the number of `peaks' is the number of connected components in the binary image. This is the rank of the $0^{th}$ homology of the binary image, denoted $rank(H_0(V_i^{-1}([0.75 \delta,\ \delta])))$ in Algorithm \ref{algo:OT}. We perform this calculation for each of the 5 frequency bands that were measured. 
Then these 5 predictions vote to produce the final prediction for the image in question. 

Algorithm \ref{algo:OT} can be compared to the method of producing a binary image directly from the source image, without first producing a sparse approximation, and counting the number of connected components in this. The accuracy rate of this naive approach is 46\% with respect to the CNN. Algorithm \ref{algo:OT} increases in accuracy to 74\% with respect to the CNN. The CNN accuracy rate is 86\% with respect to experts, but even experts agree with each other only around 70\%-75\% \cite{adamo, grasha, wei}. Given that experts are the baseline, it is impossible, without overfitting, for a computational model to do better than that. Therefore our method provides a very high performance given that no neural network training, which often takes weeks of compute time, is required. Additionally, the O.T. method is less sensitive to noise than a CNN, see \cite{zou}, which we describe and bound in Theorems \ref{thm:NoiseBound}, \ref{thm1}, and \ref{thm:robust}. Finally, with our method, variables can be simply added and removed where Equation \eqref{eq:sparseNeigh} is quickly recalculated. 

We give the confusion matrix in Table \ref{table:2} from our calculations. We test on 128 random samples. The maximum Sinkhorn iterations is 500. The cost $C_{ij}$ is chosen as the $L^2$-distance between the $i^{th}$ pixel and $j^{th}$ pixel. The $H_0$ threshold is 0.75. The initial gradient descent step size is 0.001. Wasserstein parameter $\epsilon=0.001$. Sparsity parameter $\lambda=1$. When specifying the accuracy rates in the previous paragraph we use the classification results of the CNN in \cite{perez} as the definition of the correct classification. 

\section{Conclusion}

Optimal transportation is more efficient, robust, and flexible than CNNs. We proved that optimal transportation will reconstruct sparse sources and is robust to noise. This is relevant for correcting distortions and noise in imaging which we showed for star cluster detection. Another benefit of a predictive model for star clusters is that it can produce a \emph{policy} that informs where future surveys should look for star clusters \cite{rawson, freeman}. 

\section*{Acknowledgment}
The last author thanks the Knut and Alice Wallenberg Foundation (Grant 2018-0357), BSF (Grant 2020329) and NSF (Grant DMS-1906370).

\end{document}